%% file: main.tex
\documentclass[10pt,a4paper]{article}

\input{preamble.tex}

\title{On Reinforcement Learning, Effect Handlers, and the State Monad}
\author{Ugo Dal Lago \and Francesco Gavazzo \and Alexis Ghyselen}
\date{}

\begin{document}

\maketitle

\begin{abstract}
     We study the algebraic effects and handlers as a way to support 
     decision-making abstractions in functional programs, whereas a user can 
     ask a learning algorithm to resolve choices without implementing the 
     underlying selection mechanism, and give a feedback by way of rewards. 
     Differently from some recently proposed approach to the problem based on 
     the selection monad 
     \cite{DBLP:conf/lics/AbadiP21}, we express the underlying 
     intelligence as a reinforcement learning algorithm implemented as a set of 
     handlers for some of these algebraic operations, including those for 
     choices and rewards. We show how we can in practice use algebraic 
     operations and handlers --- as available in the programming language \EFF\ 
     --- to clearly separate the learning algorithm from its environment, thus 
     allowing for a good level of modularity. We then show how the host 
     language can be taken as a $\la$-calculus with handlers, this way showing
     what the essential linguistic features are. We conclude by hinting at how 
     type and effect systems could ensure safety properties, at the same time 
     pointing at some directions for further work.
\end{abstract}


\section{Introduction}
\input{Introduction.tex}

\section{Structuring Reinforcement Learning Applications Through Functional 
Programming}
\label{s:OurApproach}
\input{OurApproach.tex}

\section{The Selection Monad, and Why it is Not an Answer}
\label{selection-monad}

\input{Selection.tex}

\section{From Practice to Theory: Effects and Handlers}
\label{s:Generalization}
\input{Generalization.tex}

\section{Type Safety}
\input{Safety.tex}

\section{Conclusion}
\input{Conclusion.tex}

\bibliography{Biblio}
\bibliographystyle{abbrvnat}


\end{document}

%% file: preamble.tex
\usepackage[T1]{fontenc}
\usepackage[utf8]{inputenc}
\usepackage{amssymb}
\usepackage{amsmath}
\usepackage{amsthm}
\usepackage{stmaryrd}
\usepackage{boxedminipage}
\usepackage{bussproofs}
\usepackage{tikz-cd}
\usetikzlibrary{positioning}
\usepackage{listings}
\usepackage{mathtools}
\usepackage{comment}
\usepackage{xspace}
\usepackage{url}
\usepackage{fullpage}
\usepackage[square,numbers]{natbib}

\newtheorem*{remark}{Remark}

\newtheorem{theorem}{Theorem}

\newcommand{\EFF}{\texttt{EFF}}
\newcommand{\OCAML}{\texttt{OCAML}}

\newcommand{\op}[1]{\mathtt{#1}}
\newcommand{\choice}{\op{choice}}

\newcommand{\la}{\lambda}

\newcommand{\key}[1]{\mathbf{#1}~}
\newcommand{\type}[1]{\mathtt{#1}}
\newcommand{\B}{\type{B}}

\newcommand{\Unit}{\type{Unit}}
\newcommand{\R}{\type{Real}}
\newcommand{\Ae}{A_E}
\newcommand{\Oe}{O_E}
\newcommand{\Arl}{A_{\RL}}
\newcommand{\Orl}{O_{\RL}}
\newcommand{\ararrow}{\rightsquigarrow}
\newcommand{\effarrow}[2]{\prescript{#1}{}{\Rightarrow}^{#2}}

\newcommand{\red}{\rightarrow}
\newcommand{\sub}[2]{[ #2 / #1]}
\newcommand{\set}[1]{\{ #1 \}}
\newcommand\midd{\; \mbox{\Large{$\mid$}}\;}
\newenvironment{framed}[0]{\begin{boxedminipage}{\linewidth}}{\end{boxedminipage}}
\newcommand{\RL}{\mathit{RL}}
\newcommand{\Abs}{\mathit{Abs}}
\newcommand{\eff}{\mathit{eff}}
\newcommand{\hide}{\mathit{hide}}
\newcommand{\Act}{\mathit{Act}}

\newcommand{\p}{\vdash}
\newcommand{\AXC}[1]{\AxiomC{#1}}
\newcommand{\UIC}[1]{\UnaryInfC{#1}}
\newcommand{\BIC}[1]{\BinaryInfC{#1}}
\newcommand{\TIC}[1]{\TrinaryInfC{#1}}
\newcommand{\DP}{\DisplayProof} 
\newcommand{\vvspace}{\vspace{3mm}}
\lstdefinestyle{eff}{
	mathescape=true
	breaklines=true,
	basicstyle=\small,
	emph={read,Read,raise,Raise,randomint,Randomint,randomfloat,Randomfloat,write,Write,print,Print,Choice,Do,Reward,Observe,Init,Op,choice,do,reward,observe,init,op}, emphstyle=\bfseries,
	emph={[2]unit,float,int},emphstyle={[2]\color{cyan!80!black}},
	keywordstyle=\color{blue!70!black},
	stringstyle=\color{orange!50!black},
	showstringspaces=false,
	morecomment=[s][\color{green!50!black}]{(*}{*)},
	morestring=[b]",
	morekeywords={begin,end,finally,continue,perform,type,use,rec,effect,handle,with,handler,let,return,fun,while,match,if,then,else,probability,in},
	literate= 
	{->}{{$\mapsto$}}2
	{bot}{{$\bot$}}1
	{odod}{{\textcolor{blue}{do}}}2
}

%% file: Introduction.tex
Machine learning is having, and will likely have more and more, a tremendous  
impact on the way \emph{computational} problems (e.g. classification or 
clustering)  are solved. Learning techniques, however, turn out to be very 
fruitful when solving \emph{control} problems, too. There, in fact, an agent's 
goal is to learn how to maximize its  reward while interacting with the 
environment rather than while computing a  mere function. From this point of 
view, the so-called \emph{reinforcement learning}  techniques 
\cite{Sutton2018:ReinforcementLearning}  are proving to be particularly 
appropriate in many contexts where  exploration and optimization have to be 
interleaved.

Prompted by that, in the past decade there has been an incredible effort to 
develop programming languages and programming language techniques oriented to 
the design of machine learning systems. The outcome of such an effort is 
well-known and gave birth to new  programming language paradigms, such as  
Bayesian~\cite{DBLP:conf/popl/Goodman13,DBLP:journals/corr/abs-1809-10756} and  
differentiable 
programming~\cite{DBLP:journals/pacmpl/AbadiP20,DBLP:journals/toplas/PearlmutterS08,DBLP:journals/lisp/SiskindP08},
as well as to the flourishing field of programming languages for 
inference.\footnote{See, e.g., the dedicated POPL Workshop LAFI 
\url{https://popl21.sigplan.org/home/lafi-2021}.}
Despite the incredible strides made, machine learning 
support from programming languages is still in its infancy
and its deliverables mostly consist of a set of 
general-purpose programming language libraries 
(such as \textsf{Theano}~\cite{Bergstra10theano}, 
\textsf{TensorFlow}~\cite{DBLP:conf/osdi/AbadiBCCDDDGIIK16}, and 
\textsf{Edward} \cite{tran2016edward,tran2017deep}, just to mention but a few) 
and 
domain-specific languages (such as \textsf{Anglican}~\cite{Tolpin/Anglican/2016}, 
\textsf{Pyro}~\cite{DBLP:journals/jmlr/BinghamCJOPKSSH19}, and 
\textsf{Stan}~\cite{JSSv076i01}, just to mention but a few).

In this work, we deal with a further programming language paradigm oriented 
to machine learning: \emph{choice-based 
programming}. The latter moves from the observation that many 
machine learning systems --- especially those pertaining the 
realm of reinforcement learning ---
can be described in terms of \emph{choices}, \emph{costs}, and 
\emph{rewards}. Prompted by that, choice-based programming languages\footnote{
	Such as \textsf{SmartChoice}~\cite{Carbune2019:PredictedVariables}.
}  extend traditional programming languages 
with high-level \emph{decision-making} abstractions 
that allow for the modular design of programs in terms of choices 
and rewards. While in a probabilistic language 
programs are structured in terms of sampling and observing 
--- leaving the actual inference process to the interpreter --- 
in a choice-based language the code is structured by specifying where a choice 
should be made and what its associated cost (or, dually, its reward) is, 
leaving the actual decision-making process to the interpreter. 

All of that considerably changes the way one writes and thinks about software, 
at the same time raising new --- and challenging --- questions
both from the point of view of programming language theory and of machine learning. 
In the former case, in fact, we have to deal with the introduction 
of decision-making abstractions, their implementation, and their semantics. 
In the latter case, instead, we have just began to realise how to tackle modularity 
of machine learning systems using programming language-based techniques. 

The latter point is the main topic of this paper. 
We move from the recent work by 
\citet{DBLP:conf/lics/AbadiP21}, where a \emph{monadic}~\cite{Moggi1998:ComputationalMonads}
approach to choice-based 
programming is developed. There, the authors show how 
the so-called 
\emph{selection monad}~\cite{DBLP:conf/cie/EscardoO10a,DBLP:conf/csl/EscardoOP11,
DBLP:journals/apal/EscardoO12} can be used to 
model (and to give semantics to) programs written 
in a choice-based language as \emph{effectful} programs, and how
it is
possible to manage choices as 
\emph{algebraic operations}~\cite{PlotkinPower/FOSSACS/01,
PlotkinPower2003:AlgebraicOperations,PlotkinPower/FOSSACS/02}, 
delegating the task of 
solving these choices to those who implement the selection mechanism.
In this work, we are concerned with further developing this idea, 
although in a different direction. In fact, even if 
we still deal with 
monadic and algebraic approaches to choice-based programming, 
we explore the use of the \emph{state} --- rather than \emph{selection} --- 
\emph{monad} in the framework of 
reinforcement learning systems. In particular, we show 
how the use of the state monad 
allows for a high level 
of modularity in the 
construction of systems based on reinforcement learning. 

The reader, at this point, may wonder \emph{why} one should consider a 
further kind of monadic programming --- viz. one based on the state monad 
--- when dealing with reinforcement learning systems .  
The next two sections are dedicated to answer this question. There, 
we will study a simple example coming from the reinforcement learning literature, 
highlighting some drawbacks of the selection monad, on the one hand, 
and the strengths of our state-based approach, on the other hand. That 
will also allow us to illustrate the practical advantages 
of functional programming techniques in reinforcement learning modelling. 

For the moment, we simply remark that at the very heart of our approach 
lies a clear separation between the part of the program dealing with the learning 
algorithm and the one that handles the interaction with the environment. 
Such a separation can be naturally structured in the form of a state monad 
with choice and reward operations acting as algebraic operations. However --- and 
here comes the main difference between our approach and the one based 
on the selection monad --- 
the action of such operations 
is \emph{not} determined by the monad (as \emph{de facto} happens when working with 
the selection monad), but it is ultimately given 
by the reinforcement learning algorithm used. 
As a consequence, one can model reinforcement learning systems modularly 
as monadic, state-based programs written using choice and reward operations, 
and then view 
reinforcement learning algorithms as algebraic interpretations of such operations. 
Concretely, that means that we can view (and implement) reinforcement learning 
algorithms as \emph{handlers}~\cite{PlotkinPretnar/handling-algebraic-effects-2013,
Bauer-Pretnar/Programming-with-algebraic-effects,
Pretnar2015:EFF} giving interpretations to choice and reward operations.
Different learning algorithms then give different handlers --- and thus 
different interpretations of choices and rewards ---
so that it is possible to instantiate the very same program 
with different learning algorithms by simply changing the way its
algebraic operations are handled. 

Summing up, the main contributions of this paper are the following:
\begin{itemize}
\item The development of a modular approach to reinforcement learning systems 
throughout functional programming language techniques. 
Among such techniques, the main role is played by the \emph{state monad} 
and its associated \emph{algebraic operations} for performing choices and 
rewards. Crucially, the use of the state monad (as opposed to the selection monad) 
allows us to consider different interpretations of choices and rewards, such 
interpretations ultimately \emph{being} the reinforcement learning algorithm 
used. We make use of \emph{handlers} to implement these interpretations in a 
modular fashion. 
\item The analysis of our approach in a core $\lambda$-calculus with 
	handlers and algebraic operations. This allows us to isolate the 
	exact features a language needs to have in order to support our 
	functional approach to reinforcement learning systems. 
\item A preliminary study of how to use semantic-based techniques (notably 
polymorphic and graded type systems) to ensure correct behaviours of 
reinforcement learning systems when implemented in a functional way. 
\end{itemize}


%% file: OurApproach.tex
To begin with, we are going to illustrate our approach on a simple running 
example, common in the reinforcement learning 
literature~\cite{Sutton2018:ReinforcementLearning}, namely the so called 
\emph{multi-armed bandit} problem, denoted MAB. In 
this setting, a gambler faces a row of $k$ slot machines, 
each of which distributes rewards in the form of winnings according 
to a probabilistic model unknown to the gambler. The gambler, it goes without 
saying, wants to maximize its gains over a fixed number of rounds. At each 
round, the gambler chooses a machine and obtains a reward depending on its 
reward distribution. On the one hand, then, the gambler wants to maximize her 
gains, but on the other she knows nothing about the inner working of the 
machines. This example, indeed, expresses the need of a trade-off between 
\emph{exploration} and \emph{optimization}: the gambler wants to play on the 
best machine as much as possible, but in order to find this best machine, she 
needs to try \emph{every} machine. Moreover, since rewards are random, playing 
once on each machine, then playing only on the one with the best observed 
reward, may not be optimal in the long run. Reinforcement learning indeed 
focuses in offering techniques which explores this trade-off in a meaningful 
way.

What we are interested in doing here is deriving a strategy for 
the agent in the aformentioned problem \emph{using} reinforcement learning 
techniques, and \emph{with the help} of functional programming. We are 
\emph{not} interested at devising new algorithms, but rather at showing 
that functional programming can help in \emph{giving structure} to programs 
which comprise both the proper learning algorithm, but also the interaction 
with the environment.  

Conceptually, it is natural to see a program solving MAB by way of reinforcement learning as structured into three parts:
\begin{itemize}
	\item
	The first one, the \textbf{environment}, serves as an interface with the 
	outside world, making it visible to the program. Here, this part comprises 
	the $k$ slot machines, possibly through a system of sensors and functions 
	querying those sensors.
    \item
    The second one, the \textbf{learner}, which provides generic reinforcement 
    learning algorithms, possibly through libraries. There are many algorithms 
    the literature offers, for example gradient learning, Q-learning, expected Sarsa or other TD-methods \cite{Sutton2018:ReinforcementLearning}.  
    \item
    The \textbf{user}, namely some code letting the environment and the learner 
    exchange information, thus acting as a bridge between the two. Here, the 
    user is supposed to turn events and data coming from the environment into a 
    form which can be understood by the learner.
\end{itemize}
One of the key ingredients of our proposal consists in structuring these three 
parts and their interaction by stipulating that the 
user proceeds by invoking some algebraic operations provided by both the 
environment and the learner. This is schematized in 
Figure~\ref{f:graphsummary}.  
\begin{figure}
	\begin{framed}
		\begin{tikzpicture}
			\node[draw, text width=3cm] (env) 
			{\centering \textbf{Environment } \\};
			\node[draw, text width = 3cm]  (user) [below right=of env] 
			{\centering \textbf{User} \\ \flushleft $O_E$ : Observations \\ 
			$A_E$ : Actions};
			\node[draw, text width = 3cm]  (user) [below right=of env] 
			{\centering \textbf{User} \\ \flushleft $O_E$ : Observations \\ 
			$A_E$ : Actions};
			\node[draw, text width = 4.7cm]  (learner) [above right=of user] 
			{\centering \textbf{Learner} \\ \flushleft $O_{\RL}$ : Abstract 
			Observations \\ $A_{\RL}$ : Abstract Actions};
			\draw[->] (env.east) to [bend left] node[below left] 
			{$\op{observe}$} (user.110) ;
			\draw[->] (user.west) to [bend left] node[above right] {$\op{do}$} 
			(env.south) ;
			\draw[->] (learner.west) to [bend right] node[below right] 
			{$\op{choice}$} (user.70) ;
			\draw[->] (user.10) to [bend right] node[above left] 
			{$\op{reward}$} (learner.250) ;
			\draw[dotted] (user.350) to [bend right] node[below right] 
			{Abstract Interface} (learner.310) ;
		\end{tikzpicture}
	\end{framed}
	\caption{Our approach}
	\label{f:graphsummary}
\end{figure}
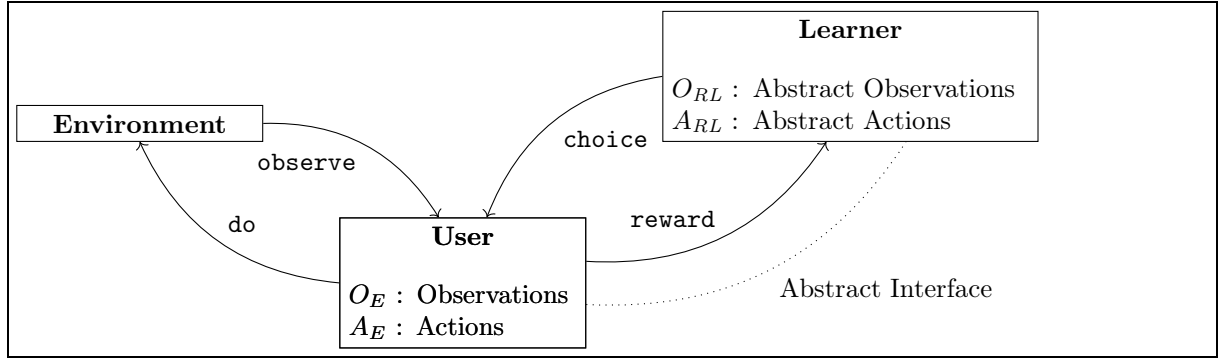
The main algebraic operations are the four mentioned in the figure, two of them 
provided by the environment and two provided by the learner. The two algebraic 
operations provided by the environment are $\op{observe}$ and $\op{do}$. 
The former,  having type $1 \ararrow O_E$ allows the user to retrieve some data 
from the environment. In our multi-armed bandit example, it corresponds to 
observing the gain on the current slot machine. The operation $\op{do} : A_E 
\ararrow 1$ executes an action in the environment. In our running example, an 
action would be the choice of one of the $k$ slot machines. Those operations 
depend the set of possible observations and the set of 
possible actions in the environment, those sets being represented by two types 
$O_E$ and $A_E$ in the user code. In our example, the observations are the 
gains, thus we can give the type $O_E = \R$. As for actions, they stand for 
choices of a slot machine, and thus $A_E = \set{1,\dots,k}$. 

Then, there are two other operations for the learner, namely $\op{choice}$ and 
$\op{reward}$. The operation $\op{choice} : 1 \ararrow A_E$ asks the learner 
to take an action for us. In our setting it means that a call to $\op{choice}$ 
returns an integer between $1$ and $k$ designing the slot machine the user 
should play on. The $\op{reward} : \type{Real} \ararrow 1 $ operation allows 
the user to give a feedback to the learner depending on its previous choice. In 
our example, the typical reward would be the gain obtained from the slot 
machine chosen by the learner. There is also an additional bridge between the 
user and the learner, that we call the \emph{abstract interface}: the learner 
may need more information than the feedback given by 
$\op{reward}$, and this information will be transmitted using the abstract 
interface by way of some additional algebraic operations invoked by the learner 
but handled by the user.

One advantage of the just described approach is to show that this learner part 
of the program, i.e., the handler for the operations $\op{choice}$ and 
$\op{reward}$, can be implemented \emph{independently} from the details of the 
environment, once and for all. One could even conceive to have a library of 
handlers, where each set of handlers corresponds to a reinforcement learning 
algorithm, in such a way that the user could choose one of these handlers for 
each of the environments he has access to, implementing the abstract interface 
but without delving into the details of the underlying algorithms. Similarly, 
implementing any new RL algorithm would boil down to just write a new set of 
handlers for  $\op{choice}$ and $\op{reward}$. 

In which \emph{order} should the user invoke the various algebraic operations?
A typical sequence of interactions would be $\op{choice} \rightarrow \op{do} 
\rightarrow \op{observe} \rightarrow \op{reward}$, meaning intuitively that it 
asks for a choice, and then does explicitly this action in the environment, 
observes the results of this action, and finally produces a reward. As we 
will soon see, it is preferable that the user indeed respects this order, e.g., 
when handling the $\op{reward}$ operation, the learner may need to observe the 
environment (abstracted by the interface) since a valuation function typically 
depends on the state reached after an action. 

In order to have a learner independent from the environment, it must not make 
explicit reference to the types $\Ae$ and $\Oe$, as we saw on MAB. Thus, we introduce two \emph{finite} sets $O_{\RL}$ and $A_{\RL}$ 
corresponding to the abstract sets of observations and actions, respectively. 
Except for their sizes, those finite sets must be totally independent form the 
environment. It means that in practice, the learner only has access to some 
elements of those sets, and should make choices in $\Arl$ depending only on the 
abstract observations from $\Orl$ it receives from the user. Typically, both 
the policy to choose actions and the learning part are built by constructing 
and updating a \emph{value function} mapping pairs in $O_{\RL} \times A_{\RL}$ to 
a reward estimation, for example an element of $\R$. 

\subsection{A Naive RL Algorithm in \EFF}

We present our approach practically with a very basic algorithm for 
reinforcement learning, which works by only remembering an evaluation function 
keeping track of \emph{the expectation} 
of the immediate rewards it receives for any of the actions, and chooses an 
action by way of a so-called $\varepsilon$-greedy policy 
\cite{Sutton2018:ReinforcementLearning}: 
with probability $(1 - \varepsilon)$, it 
makes the best possible choice based on the \emph{current} valuation, and 
with 
probability $\varepsilon$ it explores by choosing an action uniformly at 
random. To implement this algorithm, we use the language \EFF 
~\cite{Pretnar2015:EFF}, an \OCAML-based language for effects and handlers. The 
syntax of \EFF\ should be 
understandable to anyone with some basic knowledge on \OCAML, effects and 
handlers. We use lists for 
the sake of simplicity, although other kinds of data structures would enable 
better performances. The source code as well as other examples can be found in 
\cite{GhyselenEFFCode}.

\paragraph{Declaring The Abstract Interface.}
The first step towards implementing the RL algorithm consists in declaring the 
sets $O_{\RL}$ and $A_{\RL}$ together with the abstract interface which will 
allow us to recover some information on the sets $O_{\RL}$ and $A_{\RL}$:
\begin{lstlisting}[style = eff]
type rl_act
type rl_obs 
effect rl_observe : unit -> rl_obs 
effect rl_getavailableact : rl_obs -> rl_act list 
\end{lstlisting}
Here, the first effect allows the learner to observe the environment (after the 
abstraction) and to get, for each observation, the list of available actions it 
has to choose from. As we stated before, the learner does \emph{not} have 
access to the interpretation of those two effects nor to the actual types, and 
only knows that those sets of abstract actions and observations are finite, 
this being enough to implement the reinforcement learning algorithm. 

\paragraph{Handling Choices and Rewards.}
The handler makes essential use of the state monad, where the state represents 
the memory of the learner. For the specific RL algorithm we are targeting now, 
this memory consists of an element of this type:
\begin{lstlisting}[style = eff]
type memory = ((rl_obs*((rl_act*int*float) list)) list)*int*int
\end{lstlisting}
The left type of the internal memory corresponds to a \emph{value function}, for each pair of an observable and an action, we give an estimation 
of the immediate reward we can obtain. It is computed as the \emph{average} 
reward, and in order to do this average incrementally, it is common to also 
remember the number of times a choice has been made. Then, the internal memory 
also remembers the \emph{last choice made} using a pair of integers denoting 
indexes in the evaluation function (we usually denote $na$ the index of an 
action and $no$ the index of an observation). Then, we can implement the 
$\varepsilon$-greedy policy of the learner (we only describe the important 
functions and not the simple intermediate one). 
\begin{lstlisting}[style = eff]
(* This function takes as input a probability and a list of reward estimation
and returns the selected action and its index in the list *)
let greedypolicy ($\epsilon$,l) = 
  if ((randomfloat 1.) <= $\epsilon$) then 
  (* Uniform choice in this case *)
  begin 
    let na = randomint (list_length l) in
    (* Find the action with index na *)
    let a = findact l na in 
    (na,a)
  end  
  else
    (* Select the action with the maximal estimated reward *) 
    argmax l
;;
\end{lstlisting}
And with this, and some other auxiliary functions, we can define the handler 
for the basic RL algorithm. The only non-standard clause for the handler is the 
last one, starting with $\mathtt{finally}$, in this setting with a state 
monad, it should be understood as the initial state for a computation. 
\begin{lstlisting}[style = eff]
(*The first input is the probability of exploring, 
the second is the initial estimation*)
let rl_naive $\epsilon$ v = handler 
  (*declare a state monad, with the type described above*)
  | y -> (fun (_:memory) -> y) 
  | effect Choice k -> fun (l,_,_) -> 
    (*use the interface to get an observation*)
    let o = rl_observe () in 
    (* extract the index no for o, with its list of estimations q *)
    let (l',no,q) = getstateestimate l o v in 
    (* select the action a with the greedy policy *)
    let (na,a) = greedypolicy ($\epsilon$,q) in 
    (*update the estimations*)
    let l'' = updatestate l' no (fun ll -> updatechoice ll na) in 
    (*return action a, with the updated memory*)
    (continue k a) (l'',no,na) 
  | effect (Reward r) k -> fun (l,no,na) -> 
    (*update estimations for the previous choice (no,na) *)  
    let l' = updatestate l no (fun q -> updatereward q na r) in 
    (*give the new memory to the continuation*)
    (continue k ()) (l',no,na) 
  (*initial memory*)
  | finally f -> f ([],0,0) 
;;    
\end{lstlisting}
Note that we initialize the estimations lists only when we see an element of 
$O_{\RL}$ for the first time. This is standard in RL because there could be an 
extremely large set of states, and it may well be that not all of them are 
reached --- it may be better to give an initial estimation to a state only when 
we actually see this state 
(in this program, this is done by the $\mathtt{geststateestimate}$ function). 
We do not give the code of the update functions, which is anyway easy to write. 
The nice thing of this handler is that it \emph{does not depend} on the 
environment, it can of course interact with its environment (using 
$\mathtt{rl\_observe}$) but it does so in a modular way, so that this program 
can be used in \emph{any environment} in which we would like to experiment this 
(admittedly naive) algorithm.  

\subsection{The Multi-Armed Bandit in \EFF}

We show how to make use of the RL algorithm described in the previous section 
on the environment coming from MAB. The first step consists 
in declaring the types for $O_E$ and $A_E$. Since rewards are earnings, and our 
actions are nothing more than a choice of a specific machine, 
we can proceed as follows:
\begin{lstlisting}[style = eff]
type env_obs = float 
type env_act = int
\end{lstlisting}

\paragraph{Modeling the Environment}
We model MAB as a program, where we take a very simple  
distribution of rewards for the sake of the example. As stated before, the 
environment correspond to the handling of $\op{observe}$ and $\op{do}$. 
\begin{lstlisting}[style = eff]
type env_state = float 
	
(* The random reward of the machine a. 
In a real case, this reward should be obtained 
by observing the result of the slot machine *)
let getreward (a:act) = (float_of_int a) +. (randomfloat 10.) ;;
	
(*max corresponds to the number of slot machines *)
let MAB_handler max = handler
  | y -> (fun (_:env_state) -> y)
  | effect (Do a) k -> fun _ -> 
    (* When seeing a valid action a, we compute the gain for the 
    slot machine a and store the result in the memory *)
    if (a > 0) && (a <= max) then (continue k ()) (getreward a) 
    else raise "This action is not available! \n"
  (*An observation corresponds to showing the stored result *)  
  | effect Observe k -> fun r -> (continue k r) r   
  (* Initial environment, no rewards observed *)
  | finally f -> f 0.
;;
\end{lstlisting}

\paragraph{Implementing the Abstract Interface}
The abstract interface is a handler that implements the types $\Arl$ and 
$\Orl$, and handles the algebraic operations declared by the learner:
\begin{lstlisting}[style = eff]
(* Abstractions Types. The type unit for observations 
means that the learner has no information on the environment *)
type rl_obs = unit ;;
type rl_act = int ;;
effect rl_observe : rl_obs 
effect rl_getavailableact : rl_obs -> rl_act list 
(*Transform a standard observation into an abstract one*)
let abstractobs (o : env_obs) :obs = () ;;
	
(*The handler describes:
  - the abstraction of observations
  - the actions available to the learner *)
let abs_MAB max = 
let l = list_enumerate 1 (max + 1) in 
handler 
  | effect rl_observe k ->  let o = observe () in 
    continue k (abstractobs o) 
  | effect (rl_getavailableact o) k -> continue k l
;;	
\end{lstlisting}

\paragraph{The Main Program}

With this, we have everything we need to handle the four main operations. In 
order to use the learner described above, we can just open the file in which 
the RL algorithm is defined and use it as a handler, with the interface 
described above. For example, we can write the main program:

\begin{lstlisting}[style = eff]
#use "MAB_Environment.eff" ;;
#use "RL_Naive.eff";;
	
(*Multi-Armed Bandit with 6 machines *)
with (MAB_handler 6) handle 
(*Provides the interface to the learner *)
with (abs_MAB 6) handle 
(* Call the Basic RL algorithm described previously *)
with (rl_naive 0.05 10.) handle 
(*Start writing your progam with algebraic operations 
Here, we do 500 rounds *)
  let rec run n r = 
    if n = 0 then r else   
    let a = choice () in (do a);
    let r' = observe () in 
    reward r'; run (n-1) (r +. r')
  in run 500 0. 
;;
\end{lstlisting}

And the point is that if we want to use another learner we only have to load a 
different handler for the learner, and this naive learner can be used in any 
environment as long as the abstract interface is handled.

%% file: Selection.tex
To gain a better understanding of the differences between our \emph{state}-based 
approach and the \emph{selection}-based approach by \citet{DBLP:conf/lics/AbadiP21}, 
let us illustrate some of the main drawbacks exhibited by the selection monad 
when applied to MAB. To do so, let us first 
shortly recap the underlying mechanism behind such a monad.

The selection monad is defined through the functor $S(X)=(X\rightarrow R)\rightarrow 
X$, where $R$ is a (usually ordered) set of rewards. Intuitively, a computation in 
$S(X)$ takes in input a reward function $f: X \to R$ associating to each element 
in $X$ a reward in $R$ (we can think about $f(x)$ as a measure of 
the goodness of $x$), 
and chooses an element in $X$ that, intuitively, 
is optimal for $f$.
By its very definition, even if the set $R$ of rewards is a 
parameter of the selection monad, the latter does not 
have direct access to it: the only way to interact with rewards is through 
a reward function, meaning that the selection monad, by itself, 
does not handle rewards \emph{directly}. 
\citet{DBLP:conf/lics/AbadiP21} overcome this problem by 
combining the selection monad with another monad $T$ giving direct access to rewards, 
this way obtaining (a monad whose carrier is) the functor
$$
S_T(X)=(X\rightarrow R)\rightarrow T(X).
$$
Here, $T$ can for example be the $R$-based writer monad $T(X) = R\times 
X$ --- this way, rewards can be found in the first component of the result 
whenever a choice is made. One can even go beyond that and take 
$T(X) = D_\mathit{fin}(R\times X)$, where $D_\mathit{fin}$ is the finite 
distribution monad, this way modeling stochastic rewards. 

Using this strategy --- i.e. working with $S_T$ rather than with $S$ alone --- 
we can indeed model (stochastic) rewards as prescribed by the multi-armed bandit problem. 
For that, we just consider the monad 
$(X\rightarrow \mathbb{R}) \rightarrow  D_\mathit{fin}(\mathbb{R} \times X)$. 
The last ingredient needed to attempt a selection monad-based model to MAB is to give \emph{choice operations}, which ultimately 
constitute the way to construct selection computations. As the reader may guess, 
this is the most delicate point.

The selection monad comes with a choice operation $\op{choice}$  
which given two selection computations $a$ and $b$ in $S(X)$, 
returns a new computation $\op{choice}(a,b)$ belonging to $S(X)$ working as follow: 
given a reward function $f:X\rightarrow R$, 
$\op{choice}(a,b)$ passes $f$ to both $a$ and $b$, this way obtaining 
two candidate optimal elements $a(f)$ and $b(f)$, and then 
chooses between the latter on the basis of $f$. 
More precisely, given $a(f)$ and $b(f)$, we can obtain a reward for each of them 
as the elements 
$f(a(f))$ and $f(b(f))$ in $R$. Assuming $R$ to come with a binary relation $\preceq$ 
ranking rewards, 
we then let $\op{choice}(a,b)(f)$ to return the best one between $a(f)$ and $b(f)$ 
according to $\preceq$. 

At this point, we can already start perceiving the main
issue behind choice operations: to choose between two computations, we have to simulate 
\emph{both} 
of them, and then take the optimal one.
In any case, this concerns the selection monad $S$ alone. What 
about the monad $S_T$, i.e. the combination of the 
selection monad with a monad $T$ giving access to rewards? 
Here, the situation is slightly more complicated, but the basic mechanism behind 
choice operations is essentially the same one for $S$. Proceeding as for the latter, 
we obtain elements $a(f)$ and $b(f)$ belonging to 
$T(X)$: this time, however, we cannot directly apply the reward function $f$ 
on them.
The solution proposed by \citet{DBLP:conf/lics/AbadiP21} is to require 
to have a $T$-algebra $\alpha: T(R) \to R$, so that we can apply 
$T(f)$ --- rather than $f$ --- on $a(f)$ and $b(f)$, this way obtaining 
monadic rewards $T(f)(a(f))$, $T(f)(b(f))$ in $T(R)$; map them into $R$ using 
$\alpha$; and then choose the optimal one according to the order $\preceq$ on 
$\mathbb{R}$. For instance, taking real numbers as rewards, we can take, e.g.,
$\alpha: \mathbb{R} \times \mathbb{R} \to \mathbb{R}$ as real-number by addition. 

Notice that dealing with choice operations this way, we not only still have
the same problem seen for $S$ --- namely that choosing between computations
require to simulate \emph{all} of them --- but we also have to compute 
monadic applications of $f$. That is, the application of $T(f)$ to 
an element $\phi \in T(X)$ usually requires to compute the reward $f(x)$ 
for each element $x$ which is, intuitively, part of $\phi$. For instance, 
if we think about $\phi$ as a distribution, then computing $T(f)(\phi)$ 
requires to compute the reward of each element in the support of $\phi$. 
All of that does not fit well with the very essence of reinforcement learning, 
as we are going to see.  

Let us now come to reinforcement learning systems and to the multi-armed bandit example, 
highlighting the main problem 
behind the semantic mechanism of choice operations. 
First, let us notice that the selection monad 
seems to provide an adequate setting for reinforcement learning. In fact, taking 
the monad 
$(X\rightarrow \mathbb{R}) \rightarrow  D_\mathit{fin}(\mathbb{R} \times X)$,
we obtain choices and (stochastic) 
rewards, and the former can in principle depend on the latter. 
In such a setting, the choice operation is defined parametrically 
with respect to an algebra 
$\alpha: D_\mathit{fin}(\mathbb{R} \times \mathbb{R}) \to \mathbb{R}$, 
which can be naturally defined as follows: given a distribution 
$\phi$, we first apply addition to all the elements in the 
support of $\phi$, and then compute the resulting 
expectation. Given such an algebra $\alpha$, 
to perform a choice operation $\op{choose}$ between two computations 
$a$ and $b$ with a reward function $f$, we first compute the probability distributions 
$a(f)$ and $b(f)$, and 
then compute the expected reward associated to \emph{each} such a distribution, 
which means first computing the reward of \emph{each} element in its support. 

The latter passage is not reasonable from reinforcement learning perspective.  
In fact, when dealing with reinforcement learning systems, implementing choices operations this way forces us to $(i)$ simulate \emph{all} actions; $(ii)$ to obtain 
a perfect knowledge of the environment; $(iii)$ and to make (optimal) choices based on that. 
In MAB, that means $(i)$ simulating playing on \emph{each} machine; 
$(ii)$ observing the \emph{whole} distributions associated to each machine; 
$(iii)$ and then select the one with the best expectation, which in turn
requires to compute the rewards of \emph{each} element in the support of the 
distribution of a machine. All of that is simply to strong to give a reasonable model of 
reinforcement learning systems. Notice that all of that is essentially independent 
of the choice of the algebra $\alpha$, meaning that the problem lies at the very hearth 
of he selection monad (and its choice operations) rather than on the concrete way 
one aggregates monadic rewards.

%% file: Generalization.tex
In this section, we try to inject the ideas we developed in Section~\ref{s:OurApproach} into 
a paradigmatic programming language, so as to be able to isolate the features 
we deem necessary.

\subsection{A Core Language with Effects and Handlers}
For the sake of simplicity, we take the core language described in 
\cite{Pretnar2015:EFF} with additional base types corresponding to observations 
and actions and a simple type system without type effects. We will not talk 
about the details of the underlying effect, such as randomness, I/O, or 
exceptions. Those effects are obviously needed in practice, but adding them to 
the theory would be standard, so for the sake of simplicity we ignore those, 
and when we write a type $T \rightarrow T'$, this function should be understood 
as a function that can use those standard effects. 

The grammar for terms and types is in Figure~\ref{f:effgrammar}.
\begin{figure}
	\begin{framed}
		\begin{align*}
			\text{(Values) } V &::= x \midd \la x. C \midd (V,V) \midd () \midd 
			\underline{c} \midd \underline{f} \midd H \\
			\text{(Handlers) } H &::= \key{handler} \{ \key{return} x \mapsto 
			C, \op{op}_1(x;k) \mapsto C, \dots, \op{op}_n(x;k) \mapsto C \} \\
			\text{(Computations) } C &::= \key{return} V \midd \pi_i(V) \midd 
			\op{op}(V;x.C) \midd \key{let} x = C ~\key{in} C \midd V~V \midd 
			\key{with} V ~\key{handle} C \\
			\text{(Types) } T &::= \B \midd \Unit \midd T \times T \midd T 
			\rightarrow T \midd T \Rightarrow T\\
		\end{align*}
	\vspace{-30pt}
	\end{framed}
	\caption{Syntax and Types of Terms}
	\label{f:effgrammar}
\end{figure}
In other words, we work with a $\la$-calculus with pairs base types (ranged 
over by $\B$), constants and functions for those base types, denoted by 
$\underline{c}$ and 
$\underline{f}$. Any symbol $c \in \mathcal{C}$ is associated to a 
base type $\B$, while any symbol $f \in \mathcal{F}$ comes equipped with a 
function type $\B_1 \times \cdots \times \B_n \rightarrow \B$. Moreover, we 
have algebraic 
operations, each operation symbol $\op{op}$ coming with a type $\op{op} : 
T_p \ararrow T_a \in \Sigma$ where $T_p$ is the type of parameters and $T_a$ 
is the arity. In a computation $\op{op}(V;x.C)$, the variable $x$ is bound in 
$C$. Those operations are handled by handlers. An 
handler has the following form
$$
\key{handler} \{ \key{return} x \mapsto C_r, \op{op}_1(x;k) \mapsto C_1, \dots, 
\op{op}_n(x;k) \mapsto C_n \}
$$
Here, the computations $C_1,\ldots,C_n$ are pieces of code meant to handle the
corresponding algebraic operation, while $C_r$ is meant to handle a return 
clause. We use the arrow $\Rightarrow$ to denote handler types, contrary to the 
function type that uses $\rightarrow$. When one wants to use the aforementioned 
handler for the purpose of managing some algebraic operations, we do by way of 
a term in the form
$$
\key{with} V ~\key{handle} C
$$ 
in which $C$ is executed in a protected environment such that any algebraic 
operations produced by $C$ is handled by the handler $V$, provided it is one 
among those declared in it.

The typing rules for this language are given in Figure~\ref{f:efftype}. In a computation $\op{op}(V;x.C)$, $C$ can be seen as a continuation for the computation, with type $T_a \rightarrow T$, this is why in the typing of handler, the second parameter $k$ has this type. The typing rule for the handler with simple types looks like a function application, where the handler transforms a computation of type $T_2$ to a computation of type $T_1$. 

\begin{figure}
	\begin{framed}
		\begin{center}
			\AXC{}
			\UIC{$\Gamma, x : T \p x : T$}
			\DP 
			\qquad 
			\AXC{$\Gamma, x : T_1 \p C : T_2$}
			\UIC{$\Gamma \p \la x. C : T_1 \rightarrow T_2$}
			\DP 
			\qquad 
			\AXC{$\Gamma \p V_1 : T_1$}
			\AXC{$\Gamma \p V_2 : T_2$}
			\BIC{$\Gamma \p (V_1,V_2) : T_1 \times T_2$}
			\DP
			\\ 
			\vvspace
			\AXC{}
			\UIC{$\Gamma \p () : \Unit$}
			\DP
			\qquad
			\AXC{$c : \B \in \mathcal{C}$}
			\UIC{$\Gamma \p \underline{c} : \B$}
			\DP
			\qquad
			\AXC{$f : \B_1 \times \cdots \B_n \rightarrow \B \in \mathcal{F}$}
			\UIC{$\Gamma \p \underline{f} : \B_1 \times \cdots \B_n \rightarrow \B$}
			\DP
			\qquad 
			\\ 
			\vvspace
			\AXC{$\{(\op{op}_i : T^p_i \ararrow T^a_i) \in \Sigma \qquad 
			\Gamma, 
			x : T_i^p, k : T_i^a \rightarrow T_c \p C_i : T_c\}_{1 \le i \le 
			n}$}
			\AXC{$\Gamma, x : T_v \p C_r : T_c $}
			\BIC{$\Gamma \p \key{handler} \{ \key{return} x \mapsto C_r, \op{op}_1(x;k) \mapsto C_1, \dots, \op{op}_n(x;k) \mapsto C_n \} : T_v \Rightarrow T_c $}
			\DP
			\\ 
			\vvspace
			\AXC{$\Gamma \p V : T$}
			\UIC{$\Gamma \p \key{return} V : T$}
			\DP 
			\qquad 
			\AXC{$\Gamma \p V : T_1 \times T_2$}
			\UIC{$\Gamma \p \pi_i(V) : T_i$}
			\DP 
			\\ 
			\vvspace
			\AXC{$(\op{op} : T_p \ararrow T_a) \in \Sigma$}
			\AXC{$\Gamma \p V : T_p$}
			\AXC{$\Gamma, x : T_a \p C : T $}
			\TIC{$\Gamma \p \op{op}(V;x.C) : T$}
			\DP 
			\\ 
			\vvspace 
			\AXC{$\Gamma \p C_1 : T_1 $}
			\AXC{$\Gamma, x : T_1 \p C_2 : T_2 $}
			\BIC{$\Gamma \p \key{let} x = C_1 ~\key{in} C_2 : T_2$}
			\DP 
			\qquad 
			\AXC{$\Gamma \p V_1 : T_2 \rightarrow T_1$}
			\AXC{$\Gamma \p V_2 : T_2$}
			\BIC{$\Gamma \p V_1~V_2 : T_1$}
			\DP
			\\ 
			\vvspace 
			\AXC{$\Gamma \p V : T_2 \Rightarrow T_1$}
			\AXC{$\Gamma \p C : T_2$}
			\BIC{$\Gamma \p \key{with} V ~\key{handle} C : T_1$}
			\DP
		\end{center}
	\end{framed}
	\caption{Typing Rules}
	\label{f:efftype}
\end{figure}

Then, the dynamic semantics is given in Figure~\ref{f:effsemantics}. Algebraic 
operations can commute with the $\key{let}$ constructor and handlers for other 
operations. As for handlers, the return computation is handled by the return 
clause, and an algebraic operation is handled by the corresponding computation 
in the handler, where the continuation $k$ is replaced by the actual 
continuation $C$ with the same handler. 

\begin{figure}
	\begin{framed}
		\begin{center}
			\small 
			\AXC{}
			\UIC{$\pi_i(V_1,V_2) \red V_i$}
			\DP 
			\qquad 
			\AXC{$C_1 \red C_1'$}
			\UIC{$\key{let} x = C_1 ~\key{in} C_2 \red \key{let} x = C_1' ~\key{in} C_2 $}
			\DP 
			\\
			\vvspace 
			\AXC{}
			\UIC{$\key{let} x = \op{op}(V;y.C_1) ~\key{in} C_2 \red \op{op}(V;y.\key{let} x = C_1 ~\key{in} C_2) $}
			\DP 
			\qquad 
			\AXC{}
			\UIC{$\key{let} x = \key{return} V ~\key{in} C_2 \red C_2 \sub{x}{V} $}
			\DP
			\\ 
			\vvspace 
			\AXC{}
			\UIC{$(\la x. C)~V \red C \sub{x}{V}$}
			\DP
			\qquad 
			\AXC{$V = (\underline{c_1},\dots,\underline{c_n}) $}
			\AXC{$ c = f(c_1,\dots,c_n) $}
			\BIC{$\underline{f}~V \red \underline{c}$}
			\DP
			\\ 
			\vvspace 
			For the following rules, we denote $H = \key{handler} \{ \key{return} x \mapsto C_r, \op{op}_1(x;k) \mapsto C_1, \dots, \op{op}_n(x;k) \mapsto C_n \}$
			\\ 
			\vvspace 
			\AXC{$C \red C'$}
			\UIC{$\key{with} H ~\key{handle} C \red \key{with} H ~\key{handle} C'$}
			\DP 
			\\
			\vvspace 
			\AXC{$\op{op} \notin \set{\op{op}_1,\dots,\op{op}_n}$}
			\UIC{$\key{with} H ~\key{handle} \op{op}(V;x.C) \red \op{op}(V;x.\key{with} H ~\key{handle} C)$}
			\DP 
			\\ 
			\vvspace 
			\AXC{}
			\UIC{$\key{with} H ~\key{handle} \key{return} V \red C_r \sub{x}{V}$}
			\DP 
			\\ 
			\vvspace  
			\AXC{}
			\UIC{$\key{with} H ~\key{handle} \op{op}_i(V;x.C) \red C_i \sub{x}{V} \sub{k}{\la x. \key{with} H ~\key{handle} C }$}
			\DP 
		\end{center}
	\end{framed}
	\caption{Semantics}
	\label{f:effsemantics}
\end{figure}

With this, we have a complete description of a core language with effects, handlers and base types. In practice, we also want other standard constructors for booleans (if then else), lists or other data structures, but this can be easily added to the core language without any theoretical difficulties. It is then relatively standard to prove the subject reduction of this type system, proving some simple kind of safety: 

\begin{theorem}[Subject Reduction]
	If $\Gamma \p C : T$ and $C \red C'$ then $\Gamma \p C' : T$
\end{theorem}

\begin{proof}
	The proof is standard. We start by proving a weakening lemma (if $\Gamma \p C : T$ then $\Gamma, x : T' \p C : T$) and a value substitution lemma (if $\Gamma, x:T \p C : T'$ and $\Gamma \p V : T$ then $\Gamma \p C \sub{x}{V} : T' $) by induction on judgment, and then we can prove subject reduction by induction on the relation $C \red C'$. The weakening lemma is useful for the cases when an algebraic operation $\op{op}$ commute with a $\key{let}$or an $\mathbf{handle}$, and the substitution lemma is useful for substitutions, which are always for values as one can see in Figure~\ref{f:effsemantics}. 
\end{proof}

However, with those simple types it is not possible to prove that any typable computation reduces to a computation of the shape $\key{return} V$, because a non handled operation $\op{op}$ cannot be reduced. To have those kind of safety theorems, we need type effects, as we will see in the next section. 

\subsection{Setting the Stage: Types and Algebraic Operations for RL}

Let us now instantiate more clearly the set of base types and operations for our approach. We consider that: 

\begin{itemize}
	\item Base types should contain at least 
	$\type{Bool},\R,\Ae,\Oe,\Arl,\Orl$, respectively the types for booleans, 
	real number for rewards, \emph{actions}, \emph{observations} and their 
	abstract counterparts.
	\item Constants for booleans and real numbers are standard. For $\Ae, \Arl$ and $\Orl$, we consider finite sets of constants, and $\Ae, \Arl$ should have the same size. As for $\Oe$, the set of constants depends on the environment so we have no fixed choice. 
	\item We have standard functions for booleans and real numbers, and equality for all those base types. We may also have additional functions for $\Oe$ depending on the environment. 
	\item For the set of operations symbols, we suppose that we have at least the main four operations described before
	\begin{align*}
		\op{choice} &: \Unit \ararrow \Ae \\
		\op{reward} &: \R \ararrow \Unit \\ 	
		\op{observe} &: \Unit \ararrow \Oe \\
		\op{do} &: \Ae \ararrow \Unit		
	\end{align*}
	And, in order to make the abstraction more formal, we also add 
	\begin{align*}
		\op{choice_{\RL}} &: \Unit \ararrow \Arl \\
		\op{reward_{\RL}} &: \R \ararrow \Unit 
	\end{align*}
	Moreover, we also need to add effects corresponding to the abstract interface. As this interface is not fixed, we chose for the sake of the example the one we used on the naive algorithm. So we add those operations: 
	\begin{align*}
		\op{observe_{\RL}} &: \Unit \ararrow \Orl \\
		\op{getactions_{\RL}} &: \Orl \ararrow \Arl~\type{List} 
	\end{align*}
\end{itemize}

This describes all we need in order to formalize our approach. However, in practice, we do not want all the actors (environment, user and learner) to have access to all those operations at all time. A typical example is that the learner must not use $\op{observe}$ or have access to $\Oe$ as it would break abstraction, and thus modularity. As a first approach, we define subsets of this language, by defining subsets of base types, functions and algebraic operation, so that we can define clearly which actor has access to which constructors. In the next section on safety, we will formalize this with a type system. 

\begin{itemize}
	\item The whole language described above, with all base types and operations is denoted $\la_{\eff}^I$. This calculus will be used to define the interface, as an interface should be able to see both the constructors for the  environment and for the learner in order to make the bridge. 
	\item We denote by $\la_{\eff}^E$ the subset of this whole language without types and operations related to the learner (all types and operations with $\RL$ in the name, such as $\Arl$, $\op{choice}_{\RL}$, \dots). This language will be used for the main program and the handler for $\op{do}$ and $\op{observe}$, as it is basically the language with no concrete information about the learner.  
	\item Dually, we denote by $\la_{\eff}^{\RL}$ the language with \emph{only} the operations related to the learner. Also, in $\la_{\eff}^{\RL}$, we consider that we do not have any functions nor constants for $\Orl$ and $\Arl$ except equality and the operations in the abstract interface. Indeed, this language will correspond to the learner, and as explained before, we want the learner to be independent from its environment. An important point to make this possible, is that the learner should be modular (or ideally, polymorphic) in the types $\Orl$ and $\Arl$, and so having access to the constants of those types would break this principle. In particular, with this definition of $\la_{\eff}^{\RL}$, changing the size of the finite sets $\Orl$ and $\Arl$ does not modify the language, so whatever the size may be, the learner still uses the same language. 
\end{itemize} 

\subsection{RL Algorithms and Environments as Handlers}

We start with the main program. As showed in the previous section, the main program is just an usual functional program that has access to the types $\Oe$ and $\Ae$ as well as the four main operations $\op{choice}$, $\op{reward}$, $\op{observe}$ and $\op{do}$. Thus, it corresponds to the language we denote $\la_{\eff}^E$. This is the main focus of our work, to make it possible to program within this language. In order to do this we still need to handle those four operations, and for this we design several handlers. 

\subsubsection{The Learner}
The learner is a handler with the state monad for the operations $\op{choice}_{\RL}$ and $\op{reward}_{\RL}$, written in $\la_{\eff}^{\RL}$. So, formally, the learner is a handler 
$$H_{\RL} \equiv \key{handler} \{ \key{return} x \mapsto \la m. \key{return} (m,x), \op{choice}_{\RL}(k) \mapsto C_c, \op{reward}_{\RL} (r;k) \mapsto C_r \}$$
such that 
$$ k : \Arl \rightarrow S_M(T) \p C_c : S_M(T) \qquad r : \R, k : \Unit \rightarrow S_M(T) \p C_r : S_M(T) $$
where $T$ is any type and $S_M(T) \equiv M \rightarrow (M \times T)$ is the type for the state monad with state $M$, that represents the \emph{memory} of the learner. With this, we obtain a handler $H_{\RL}$ with type $T \Rightarrow S_M(T)$ for any $T$.  

In this handler, we can indeed encode a RL algorithm, as we did for the naive algorithm for MAB, because, informally: 
\begin{itemize}
	\item The memory $M$ can contain a value function, associating an heuristic to pairs of states and actions ($\Orl \times \Arl$). It can also be used to log information if needed, typically the previous choice, the number of times choice was called \dots 
	\item The term for choice $C_c$ has a link with the \emph{policy} of the learner. Indeed, a policy can be seen as a function of type $M \rightarrow \Arl \times M$, where, from an internal memory of the learner, we chose an action in $\Arl$ and we can also modify the memory if needed, for example to log this choice. With this policy, it is easy to obtain a computation $C_c$ with the type described above by composing with the continuation $k$. 
	\item The term for $C_r$ has a link with the update of the value function after a choice. Indeed, learner usually modify their value function after a choice (or a sequence of choices). This can be seen as an update function of type $\R \times M \rightarrow M$, when we modify the memory (and thus the value function) according to the reward. Here, using a log in the memory can come in handy for most reinforcement learning algorithm to remember which choice is being rewarded. It is then easy to see that, from this update function, we can obtain a computation $C_r$ with the type described above by composing with the continuation $k$. 
\end{itemize}

However, the learner will also need additional information from the environment, typically the current observable state or the list of available actions for this state, that is why we can also use the operations of the abstract interface in the computations $C_c$ and $C_r$. 

\subsubsection{Hiding the Memory of the Learner}

The concept of this handler for the learner, which is typically a handler with the state monad, is standard but it is not very practical. Indeed, in the type $T \Rightarrow S_M(T)$ we can see that the handled computation needs an initial memory, and the final memory is returned at the end of the computation. However, this handled computation should be done in the main program, and the user cannot provide an initial memory since it is not supposed to know the actual type of memory. Similarly, there is no reason for the user to have direct access to the memory of the learner, so this type $M$ should be hidden in the computation type. So, in practice, we want a handler for the learner of type $T \Rightarrow T$ for any $T$. Fortunately, it is possible to do this from the previous handler, and it is a standard way to make the state invisible. Suppose that the learner provides an initial memory $m_i$. Then, it can define the following handler (with the empty set of handled operations):
$$H_{\hide} \equiv \key{handler} \set{\key{return} f \mapsto \key{let} x = f~m_i ~\key{in} \pi_2(x)} $$
with type $S_M(T) \Rightarrow T$ for any $T$. So, by composing this handler with the previous one, we obtain a handler of type $T \Rightarrow T$, and the memory becomes totally hidden from the user. This construction is a way to mimic the $\key{finally}$clause of the \EFF\ language, that we used in Section~\ref{s:OurApproach}, using the standard syntax of handler.

\subsubsection{The Interface}
As the previous handler uses additional algebraic operations (the one from the interface), we need to handle them. Also, the previous handler is for the operations $\op{choice}_{\RL}$ and $\op{reward}_{\RL}$ and we need to make the bridge between those operations and the one for the main program: $\op{choice}$ and $\op{reward}$. We do this by defining two handlers in $\la_{\eff}^I$.

The first handler is a simple one, mainly abstracting the set of actions. For this, we need to define a bijection $f : \Arl \rightarrow \Ae$, which is easy to do as they are both finite sets with the same size. Then, the handler for abstracting actions is given by: 
\begin{align*}
	H_{\Act} \equiv \key{handler} \{ \key{return} x &\mapsto \key{return} x, \\
	\op{choice}(k) &\mapsto \op{choice}_{\RL}(();x. k~(f~x)) \\
	\op{reward}(r;k) &\mapsto \op{reward}_{\RL}(r;x. k~x)  \}
\end{align*}
With this handler $H_{\Act}$, with type $T \Rightarrow T$ for any $T$, we can go from the operations from the main program to the operation for the learner. And now the only thing to do in order to successfully use the handler defined by the learner is to define the handler for abstract interface. With the interface we defined for this example, the two operations $\op{observe_{\RL}}$ and
$\op{getactions_{\RL}}$, then the interface handler would look like: 
$$H_{I} \equiv \key{handler} \{ \key{return} x \mapsto \key{return} x, \op{observe}_{\RL}(k) \mapsto C_o, \op{getactions}_{\RL} (o;k) \mapsto C_a \}$$
with 
$$ k : \Orl \rightarrow T \p C_o : T \qquad o : \Orl, k : (\Arl~\type{List}) \rightarrow T \p C_a : T$$
so that the handler $H_I$ has type $T \Rightarrow T$ for any $T$. 
Those computations may use the $\op{observe}$ operation, and so with this handler, that can depend on the environment, we can handle the computations coming from the handler for the learner $H_{\RL}$. Now, only two operations remain to be handled $\op{do}$ and $\op{observe}$

\subsubsection{The Environment}
To handle the environment, we only need the types and functions of $\la_{\eff}^E$, without $\op{reward}$ and $\op{choice}$. The handler for the environment should have the shape:
$$H_{E} \equiv \key{handler} \{ \key{return} x \mapsto C_r, \op{observe}(k) \mapsto C_o, \op{do} (a;k) \mapsto C_a \}$$
such that this handler is typable with type $T \Rightarrow F(T)$ for any type $T$ where $F(T)$ a transformation of $T$. The actual computations for this handler depend strongly on the environment and so we cannot give additional information for the general case. However, in order to illustrate this handler, we show how to implement it in the case where we have a specific model of the environment.

Suppose that we can model the environment by a type $E$ and two functions: $Next_E : \Ae \times E \rightarrow E$ and  $Observe_E : E \rightarrow O$. This may seem ad-hoc but it is in fact close to a Markov Decision Process which is a common model for the environment in RL algorithm \cite{Sutton2018:ReinforcementLearning}. Indeed, in this case the $Observe_E$ function corresponds to observing a reward and the current state of the Markov Decision Process, and the $Next_E$ function corresponds to moving in the MDP after an action in $\Ae$. With those functions, we can define the following handler for the environment:
\begin{align*}
	H_{E} \equiv \key{handler} \{ \key{return} x &\mapsto \la e. \key{return} (e,x), \\
	\op{observe}(k) &\mapsto \la e. k~(Observe_E~e)~e   \\
	\op{do}(a;k) &\mapsto k~()~(Next_E~(a,e))  \}
\end{align*}
with type $T \Rightarrow S_E(T)$ for any $T$. As we saw with the learner, it is possible in this case to hide the type $E$ for the main program. This is what we did for example in the description of the MAB environment in Section~\ref{s:OurApproach}. 
\subsubsection{The Main Program with Handlers}

And now, we can interpret all the four main operations for the main program. Thus, given a computation $C$ in $\la_{\eff}^E$, we can handle all those operations with the computation: 
$$\key{with} H_E ~\key{handle} (\key{with} H_I ~\key{handle} (\key{with} H_{\RL} ~\key{handle} (\key{with} H_{\Act} ~\key{handle} C)))$$
With this composition, that we could see in the main program of Section~\ref{s:OurApproach}, we obtain a handler of type $T \Rightarrow F(T)$ if the learner hides its memory as explained before. And, with a complete model of the environment, by hiding we can then obtain a type $T \Rightarrow T$. 

%% file: Safety.tex
In the previous section, we have introduced a core simply-typed calculus and used it 
to implement our running example. The choice of the language 
was driven by a simple goal: highlighting in the simplest way 
the \emph{essential} features needed to implement our functional approach 
to reinforcement learning systems. 
The price we have to pay for such a simplicity is the lack of 
several desirable guarantees on program behavior. In fact, functional 
programming languages usually come endowed with expressive type systems ---
such as 
polymorphic~\cite{Reynolds/Logical-relations/1983}, 
linear~\cite{DBLP:conf/ifip2/Wadler90,DBLP:conf/lics/BentonW96,DBLP:conf/esop/GhicaS14}, 
and graded~\cite{Orchard:2019:QPR:3352468.3341714} type systems --- 
ensuring the validity of nontrivial program properties at static time. 
Due to its simple type discipline, our calculus offers only but a few guarantees 
on program correctness, especially if one takes into account that the simplicity of 
the type system is not reflected at the operational level, 
which is instead characterised by highly expressive constructs, such as algebraic 
operations and handlers.

In light of that, it is desirable to strengthen the power of the type 
system defined in the previous section. Clearly, there are several possible 
extensions we may look for, and it thus natural to ask what path we should choose. 
In this section, we first identify three program properties well-known in the field
functional programming  that we believe to be particularly relevant when modelling 
reinforcement learning systems functionally, and then outline how such properties 
could ensured by way of expressive type systems. Let us begin with the 
target program properties. 

\begin{enumerate}
\item \textbf{Locality of Operations}.
  As a first property, we would like to ensure (families of) algebraic operations 
  to be used only in 
  \emph{specific parts of programs}. In MAB, for instance, we would like 
  the code describing the environment, i.e. the slot machines, not to be able to 
  perform the $\choice$ operation. 
\item \textbf{Polymorphism}.
  Secondly, we would like to ensure code to be as modular as possible, 
  this way enhancing its (re)usability. Concretely, that means having 
  some form of \emph{polymorphism} at disposal. This way, 
  algorithms may be written just once, the same code being called in many 
  different environments, possibly within the same program. In other words, the 
  learner should use the types $\Orl$ and $\Arl$ in a restricted way, the only 
  relevant information about them being that they are \emph{finite} types, whose values 
  can thus be enumerated. As an example, the naive algorithm we have presented in 
  Section~\ref{s:OurApproach} could be used in any context. In the particular case of MAB, 
  we have defined $\Orl$ as the unit type $\Unit$. However, we could very well replaced 
  $\Unit$ with any other \emph{finite} type, meaning that our strategy 
  scales to environments with more information. 
\item \textbf{Linearity and Order}.
  Finally, we would like force algebraic operations to be 
  performed \emph{in a specified order}, this way ensuring the learner to 
  have all the information it needs. Let us clarify this point with an example. 
  In the main program presented in 
  Section~\ref{s:OurApproach}, we see that the flow of information follows a certain logic: for each iteration, we make a choice; we perform such a choice in the environment; 
  we observe its results; and, finally, we give the reward. Such a logic is reflected 
  in a specific execution order of algebraic operation, and breaking such an order 
  may lead the learner to obtain false information. 
\end{enumerate}
 
Having isolated our target properties, we spend the rest of this section outlining 
some possible ways to force such properties by means of type systems. 
In particular, we focus on the use of type and effect 
systems~\cite{DBLP:conf/birthday/NielsonN99}, 
polymorphism~\cite{Reynolds/Logical-relations/1983}, and 
graded modal types~\cite{Orchard:2019:QPR:3352468.3341714}.

\subsection{Algebraic Operations and Effect Typing}

Type and effect systems~\cite{DBLP:conf/birthday/NielsonN99} 
endowed traditional type systems with annotations giving information 
on what effects are produced during program execution. 
For our purposes, we can build upon well-known effect typing systems keeping 
track of which operations are handled during a computation~\cite{KammarICFP2013:Handlers}. 
Extending the simple type system of Section~\ref{s:Generalization} in this way, 
we can then ensure well-typed programs to handle all their algebraic operations.   

We now define effect typing for our core calculus and 
show how to take advantage of such a typing in the context of a 
reinforcement learning problem. 
We define an effect signature as a collection of operations and 
annotate the type of handlers with those signatures. Formally, leave the 
the syntax of terms as in Figure~\ref{f:effgrammar}, but we replace 
the one of types as follows: 
\begin{align*}
	T &::= \B \midd \Unit \midd T \times T \midd T \rightarrow^E T \midd T \effarrow{E}{E} T \\
	E &::= \set{\op{op}: T \ararrow T} \sqcup E \midd \emptyset
\end{align*}

Typing judgments for computations now are of the form $\Gamma 
\p_E C : T$, with the informal reading that that $C$ has type $T$ in a context where the 
only \emph{unhandled} algebraic operations $C$ can possibly perform 
are included in $E$. Typing judgments for values, instead, 
remain the same (i.e. $\Gamma \p V : T$), as values not being executed, 
they cannot perform any algebraic operation at all.
The inference system for our new typing judgment is given in Figure~\ref{f:efftypeeffects}. 

\begin{figure}
	\begin{framed}
		\begin{center}
			\AXC{}
			\UIC{$\Gamma, x : T \p x : T$}
			\DP 
			\qquad 
			\AXC{$\Gamma, x : T_1 \p_E C : T_2$}
			\UIC{$\Gamma \p \la x. C : T_1 \rightarrow^E T_2$}
			\DP 
			\qquad 
			\AXC{$\Gamma \p V_1 : T_1$}
			\AXC{$\Gamma \p V_2 : T_2$}
			\BIC{$\Gamma \p (V_1,V_2) : T_1 \times T_2$}
			\DP
			\\ 
			\vvspace
			\AXC{}
			\UIC{$\Gamma \p () : \Unit$}
			\DP
			\qquad
			\AXC{$c : \B \in \mathcal{C}$}
			\UIC{$\Gamma \p \underline{c} : \B$}
			\DP
			\qquad
			\AXC{$f : \B_1 \times \cdots \B_n \rightarrow \B \in \mathcal{F}$}
			\UIC{$\Gamma \p \underline{f} : \B_1 \times \cdots \B_n \rightarrow \B$}
			\DP
			\qquad 
			\\ 
			\vvspace
			\AXC{$E_1 = \set{\op{op}_i : T^p_i \ararrow T^a_i \mid 1 \le i \le n} \sqcup E_f \qquad E_2 = E_2' \sqcup E_f$}
			\noLine
			\UIC{$(\Gamma, x : T_i^p, k : T_i^a \rightarrow^{E_2} T_c \p_{E_2} C_i : T_c)_{1 \le i \le n} \qquad \Gamma, x : T_v \p_{E_2} C_r : T_c $}
			\UIC{$\Gamma \p \key{handler} \{ \key{return} x \mapsto C_r, \op{op}_1(x;k) \mapsto C_1, \dots, \op{op}_n(x;k) \mapsto C_n \} : T_v \effarrow{E_1}{E_2} T_c $}
			\DP
			\\ 
			\vvspace
			\AXC{$\Gamma \p V : T$}
			\UIC{$\Gamma \p_E \key{return} V : T$}
			\DP 
			\qquad 
			\AXC{$\Gamma \p V : T_1 \times T_2$}
			\UIC{$\Gamma \p_E \pi_i(V) : T_i$}
			\DP 
			\\ 
			\vvspace
			\AXC{$(\op{op} : T_p \ararrow T_a) \in E$}
			\AXC{$\Gamma \p V : T_p$}
			\AXC{$\Gamma, x : T_a \p_E C : T $}
			\TIC{$\Gamma \p_E \op{op}(V;x.C) : T$}
			\DP 
			\\ 
			\vvspace 
			\AXC{$\Gamma \p_E C_1 : T_1 $}
			\AXC{$\Gamma, x : T_1 \p_E C_2 : T_2 $}
			\BIC{$\Gamma \p_E \key{let} x = C_1 ~\key{in} C_2 : T_2$}
			\DP 
			\qquad 
			\AXC{$\Gamma \p V_1 : T_2 \rightarrow^E T_1$}
			\AXC{$\Gamma \p V_2 : T_2$}
			\BIC{$\Gamma \p_E V_1~V_2 : T_1$}
			\DP
			\\ 
			\vvspace 
			\AXC{$\Gamma \p V : T_2 \effarrow{E_1}{E_2} T_1$}
			\AXC{$\Gamma \p_{E_1} C : T_2$}
			\BIC{$\Gamma \p_{E_2} \key{with} V ~\key{handle} C : T_1$}
			\DP
		\end{center}
	\end{framed}
	\caption{Typing Rules with Type Effects}
	\label{f:efftypeeffects}
\end{figure}

The most important rule in Figure~\ref{f:efftypeeffects} is the one for handlers 
(cf. the typing rule for open handlers by \citet{KammarICFP2013:Handlers}): it
states that if a handler can handle a subset of the available operations 
(denoted by $\set{\op{op}_i : T^p_i \ararrow T^a_i \mid 1 \le i \le n}$ 
in Figure~\ref{f:efftypeeffects}) 
and introduce some new operations in their computations ($E_2'$), 
then the final set of free operations contains the unhandled ones in the original set ($E_f$) 
together with the new operations introduced by the handler. 

This new type system satisfies better safety properties than the previous one. In particular, it enjoys preservation (as the type system of previous section) 
and progress. 

\begin{theorem} 
Progress and preservation hold for the type system in Figure~\ref{f:efftypeeffects}. 
That is:
	\begin{enumerate}
		\item \emph{Preservation.}
			If $\Gamma \p_E C : T$ and $C \red C'$ then $\Gamma \p_E C' : T$.
		\item \emph{Progress.}
			Suppose that all functions for base types are total for the constants in 
			the language. If $\p_E C : T$, then either there exists $C'$ such that $C 
			\red C'$, or $C$ has the shape $\key{return} V$ for some $V$, or $C$ has 
			the shape $\op{op}(V,x.C')$ with $\op{op} \in E$.   
	\end{enumerate}
\end{theorem}
\begin{proof}
	We prove progress by induction on $C$. The use of typed effects is important to keep track of what are exactly the set of operations that can appear, and the commutation of $\op{op}$ with $\key{let}$and $\key{handler}$is also essential in this proof. The hypothesis on base type functions ensures that an application of a base type function can always be reduced. The case of $\key{with} H ~\key{handle} C$ is the most interesting one, by induction hypothesis and typing, there are four cases for the computation $C$, either it can be reduced to $C'$, and then we use the first rule for handlers of Figure~\ref{f:effsemantics}, either it is an operation $\op{op} \in E_f$ and we use the second rule, either it is a $\key{return} V$ and we use the third rule, or it is an operation $\op{op} \in E_1/E_f$ and we use the fourth rule. 
\end{proof}

\begin{remark}
Even if we have focused on progress and preservation, it is worth mentioning that 
that termination of the 
reduction relation can be proved using standard methods~\cite{KammarICFP2013:Handlers}.
\end{remark}

We now have all the ingredients needed to encode the informal policies 
we could build the MAB example upon. Let us define the following sets of 
effects:
\begin{align*}
	E_E &= \set{\op{observe}: \Unit \ararrow \Oe ~;~ \op{do} : \Ae \ararrow \Unit} \\ 
	E_{\RL} &= \set{\op{choice}: \Unit \ararrow \Ae ~;~ \op{reward} : \R \ararrow \Unit} \\
	E^{\Abs}_{\RL} &= \set{\op{choice}_{\RL}: \Unit \ararrow \Arl ~;~ \op{reward}_{\RL} : \R \ararrow \Unit} \\
	E^{\Abs}_I &= \set{\op{observe}_{\RL}: \Unit \ararrow \Orl ~;~ \op{getactions} : \Orl \ararrow \Arl~\type{List} ~;~ \cdots} 
\end{align*}
These are, respectively, the set of effects for the environment, the learner seen by the environment, the learner seen by the learner, and the interface used by the learner. 

\paragraph{The Learner}
The handler $H_{\RL}$ for the learner should have the type 
$$ \p H_{\RL}  : T \effarrow{E_{\RL}^{\Abs}}{E_I^{\Abs}} S_M(T), $$
where $T$ is any type and $S_M(T) \equiv M \rightarrow M \times T$. This means 
that the learner handles the operations $\op{choice}_{\RL}$ and 
$\op{reward}_{\RL}$ using only the operations in the interface $E^{\Abs}_I$. 
Thus, with this type, we ensure that the learner does not have a direct access 
to the environment, and it only sees effects for the abstract types $\Arl$ and 
$\Orl$ and not the concrete types $\Oe$ and $\Ae$. Then, with hiding, we can 
obtain a handler of type $T \effarrow{E_{\RL}^{\Abs}}{E_I^{\Abs}} T$ as long as 
the learner provides an initial memory. Notice that in this hiding handler of 
type $T \effarrow{E_I^{\Abs}}{E_I^{\Abs}} T$, the initial memory can depend on 
the abstract interface, this being useful in those cases in which some 
parameters in the initial memory have to be exposed.

\paragraph{The Interface}
The interface is separated into two handlers: $H_{\Act}$ and $H_I$. This first 
handler is only a bridge between abstract actions and concrete actions. It is 
not difficult to see that the former should have the type:
$$ \p H_{\Act} : T \effarrow{E_{\RL}}{E_{\RL}^{\Abs}} T, $$
for any type $T$, thus handling the operations $\op{choice}$ and $\op{reward}$ seen from the exterior to the internal one of the learner. As for the interface, it should be a bridge between the environment and learner abstracted by the user, so it should have the type: 
$$ \p H_{I} : T \effarrow{E_I^{\Abs} \sqcup E_E}{E_E} T,$$ 
where the fixed set of operations ($E_f$ in the rule) would be $E_E$. 

\paragraph{The Environment}
As for the environment, the handler only has to handle the two operations $\op{observe}$ and $\op{do}$, without using the learner nor the interface. So, it should have the type:
$$\p H_E : T \effarrow{E_E}{\emptyset} F(T).$$

\subsubsection{Composition of Handlers and Main Program}
Now that we have types for each of the four handlers needed, we compose them. An important point of our effect typing is that it supports weakening: in the handler rule, we can always take a larger set $E_f$, so that we are licensed to use the handler in different contexts.\footnote{This is one of the main reasons why this rule is very useful to achieve  polymorphism.} In our case, it means that we can give the following types to handlers: 
\begin{align*}
	\p H_{\RL} &: T \effarrow{E_{\RL}^{\Abs} \sqcup E_E}{E_I^{\Abs} \sqcup E_E} T \\
	\p H_{\Act} &: T \effarrow{E_{\RL} \sqcup E_E}{E_{\RL}^{\Abs} \sqcup E_E} T \\
	\p H_I &: T \effarrow{E_I^{\Abs} \sqcup E_E}{E_E} T \\
	\p H_E &: T \effarrow{E_E}{\emptyset} F(T).
\end{align*}
Consequently, we can compose them to obtain: 
$$ H_E \circ H_I \circ H_{\RL} \circ H_{\Act} : T \effarrow{E_{\RL} \sqcup E_E}{\emptyset} F(T).$$
When we have a model of the environment, with hiding, we have 
$$ H_E \circ H_I \circ H_{\RL} \circ H_{\Act} : T \effarrow{E_{\RL} \sqcup E_E}{\emptyset} T. $$
This shows that the main program can use the four main effects, as expected. Without hiding, we would have a type which is essentially the composition of the two state monads on $E$ and $M$, meaning that the program should be understood in a context with both an environment and a memory for the learner. 

\subsection{Polymorphism}
The typing rule for handlers allows us to assign handlers arbitrary types $T$. 
To ensure such types not to be inspected by programs, and to ensure uniqueness of
 typing derivation of handlers, it is natural to extend our type system with polymorphic types. Moreover, the base types we presented for the language, $\Arl, \Orl, \Ae$ and $\Oe$ all depend on the actual environment, whereas we expect the handler $H_{\RL}$ for the learner not to do so. Consequently, we may rely on polymorphism to enforce  
 $\Arl$ and $\Orl$ not to be inspected, this way giving a unique polymorphic type to the handler $H_{\RL}$ and thus ensuring the latter to be usable in any environment.  

Languages, however, can be polymorphic in may ways: especially in presence of effects.
We now outline what kind of polymorphism is needed for our purposes. 
As a first requirement, we need to be able to declare polymorphic handlers; write their implementation once and for all; and then use them in different contexts. To do that, 
we believe having handlers as first-class citizen of the language, as we presented in Figure~\ref{f:effgrammar}, is necessary. Secondly, we also want to declare polymorphic effects, 
so to give a sense to those polymorphic handlers. 

We take as an inspiration the language by \citet{Biernackietal2019:AbstractingEffects}. However, since the aforementioned language does not consider handlers as first-class citizens, we cannot directly rely on that. 
In the rest of this section, we informally outline a possible polymorphic type assignment 
for our language, leaving its formal definition and analysis (such as type soundness) 
as future work.  

We consider different kinds, starting from the base kinds of
types, effects, and rows, and building complex kinds using a functional arrow.
Formally, kinds are generated by the following grammar:
$$\kappa:= T \midd E \midd R \midd \kappa \rightarrow \kappa.$$
 A row $\rho$ is a sequence of effects $\langle E_1 \mid \rho' \rangle$, with intuitively the same meaning as effect signature in our last type system, the main difference being that we can have polymorphism on row, and an effect can appear several time in a row (but possibly with different instantiation of their polymorphic types). 
We would like to write, as in \cite{Biernackietal2019:AbstractingEffects}, the following computation: 
$$\key{effect} E_{\RL}^{\Abs} = \forall \alpha_A :: T. \set{\op{choice}_{\RL} : \Unit \ararrow \alpha_A ; \op{reward}_{\RL} : \R \ararrow \Unit} ~\key{in} \cdots $$    
declaring a new effect, called $E_{\RL}^{\Abs}$, waiting for a type $\alpha_A$ representing the type (of kind $T$) of actions, and then declare the types of those two algebraic operations. The kind of $E_{\RL}^{\Abs}$ would then be $T \rightarrow E$, so given a type of kind $T$ this indeed gives an effect of kind $E$. Similarly, the abstract interface would be declared with:  
$$\key{effect} E_{I}^{\Abs} = \forall \alpha_A :: T, \alpha_O :: T. \set{\op{observe}_{\RL} : \Unit \ararrow \alpha_O ; \op{getactions}_{\RL} : \alpha_O \ararrow \alpha_A~\type{List}; \cdots} $$  
And then, in this context, the type of the handler for the learner, $H_{\RL}$ would be, after hiding:
$$ H_{\RL} : \forall \alpha_A :: T, \alpha_O :: T, \alpha_T :: T, \rho :: R. \alpha_T \effarrow{\langle E_{\RL}^{\Abs}~\alpha_A \mid \rho \rangle }{\langle E_{I}^{\Abs}~\alpha_A~\alpha_O \mid \rho \rangle} \alpha_T $$
meaning that for any type of actions $\alpha_A$, any type of observations $\alpha_O$, any type of computation $\alpha_T$ and any effect environment $\rho$, the handler for the learner has an identity type on the computation, handling the operations in $E_{\RL}^{\Abs}~\alpha_A$ and introducing the new operations of $E_I^{\Abs}~\alpha_A~\alpha_O$. In practice, we would like $\alpha_A$ and $\alpha_O$ to always represent base types, and especially finite sets, so we would need some more restriction on this polymorphism. But informally, this is the shape of type we would like for the handler written by the learner, as it ensures that it can be used in any environment without modifications. Notice that in order to do this, we want handlers as first-class citizens but we do not need the full power of polymorphism since the algebraic operations themselves do not need to be polymorphic. This kind of polymorphism that we need differs slightly from the polymorphism mainly studied in the literature: usually, it allows the programmer to use different handlers for the same operation, however what we want in our case is to use a unique handler for different contexts. 
\subsection{Imposing an Order on Operations}
Our last property of interest focuses on linearity and order of algebraic operations. 
For instance, we would like to ensure that when a call to $\op{choice}$ is made, it is always followed by a $\op{do}$ and a $\op{reward}$ ($\op{observe}$ can be useful to give the reward but it is not mandatory as it should have no effect on the environment). 
We now outline how such a goal can be achieved in presence of algebraic operations 
relying on suitable graded type systems \cite{Kastumata2014:ParametricEffects}. 
Such an approach, however, is not readily extendable to effect handlers. As far as we know, graded types have not been extended to handlers. 

Informally, we define a monoid on algebraic operations and associate each computation 
to an element of this monoid, the latter giving information on 
the operations executed during program evaluation. 
Sequential composition of computations corresponds to monoid multiplication, basic computations to the unit of the monoid, and algebraic operations to the corresponding 
elements of the monoid. A typing rule (for computations) then has shape 
$\Gamma \p_E C : T ; m$, where $m$ is an element of the monoid with, for example, the following rules: 
\begin{center}
	\AXC{$\Gamma \p V : T$}
	\UIC{$\Gamma \p_E \key{return} V : T ; 1$}
	\DP
	\qquad 
	\AXC{$\Gamma \p_E C_1 : T_1 ; m_1$}
	\AXC{$\Gamma, x : T_1 \p_E C_2 : T_2 ; m_2$}
	\BIC{$\Gamma \p_E \key{let} x=C_1 ~\key{in} C_2 : T_2 : m_1 \cdot m_2 $}
	\DP
	\\ 
	\vvspace 
	\AXC{$(\op{op} : T_p \ararrow T_a) \in E$}
	\AXC{$\Gamma \p V : T_p$}
	\AXC{$\Gamma, x : T_a \p_E C : T ; m$}
	\TIC{$\Gamma \p_E \op{op}(V;x.C) : T ; m_{\op{op}} \cdot m$}
	\DP 
\end{center}
where $1$ is the unit of the monoid, $m_1 \cdot m_2$ is monoid multiplication and $m_{\op{op}}$ is the element of the monoid corresponding to the operation $\op{op}$. 

In our case, the monoid would be the free monoid on the three element associated to the algebraic operations $\op{choice}, \op{do}$ and $\op{reward}$, with $\op{observe}$ confounded with the unit of the monoid, with the following equation
\begin{align*}
	&m_\op{choice} \cdot m_\op{do} \cdot m_\op{reward} = 1 
\end{align*}
and then, we would like the main program to be associated to an element of this monoid with the shape $m_\op{do}^n$. With this, we can ensure that, if we forget about $\op{observe}$ since it has no effect at all, then any $\op{choice}$ is always followed by exactly one $\op{do}$ and one $\op{reward}$. Moreover, only the $\op{do}$ operation can be done outside of this loop, because essentially the user can make choices without calling the learner, and for the point of view of the learner then the user would just be a part of the transition function of the environment.  

However, this method, which is standard for effects, does not generalise well in presence of handlers. Indeed, we need to take into account the potentially new operations introduced by a handler, and intuitively the handler would then represent a map from the new operations to an element of the monoid. For a very simple example, consider the operation $\op{choicedoandobserve} : 1 \ararrow \Ae \times \Oe$ with the handler 
$$\key{handler} \set{\key{return} x \mapsto x ; \op{choicedoandobserve}(k) \mapsto \op{choice}(a. \op{do}(a;\_. \op{observe}(o. k~(a,o))))}$$
doing the three operations $\op{choice}$, $\op{do}$ and $\op{observe}$ in sequence. Then, this handler should be typed with the information that the new operation $\op{choicedoandobserve}$ is associated to the element $m_\op{choice} \cdot m_\op{do} \cdot m_\op{observe}$. 

We believe such a type system should be feasible, but we have no concrete formalization of this. So, as for polymorphism, those kind of type system would be desirable for safety but we leave the formalization to future work.

\section{Related Work}
The interaction between programming language theory and machine 
learning is an active and flourishing research area. 
Arguably, the current, most well-established products of such an 
interaction are the so-called Bayesian~\cite{tran2016edward,tran2017deep,
Tolpin/Anglican/2016,DBLP:journals/jmlr/BinghamCJOPKSSH19,JSSv076i01}
and differentiable~\cite{Bergstra10theano,DBLP:conf/osdi/AbadiBCCDDDGIIK16}
programming languages, and their associated 
theories~\cite{DBLP:conf/popl/Goodman13,DBLP:journals/corr/abs-1809-10756,DBLP:journals/pacmpl/AbadiP20,DBLP:journals/toplas/PearlmutterS08,DBLP:journals/lisp/SiskindP08}. 

Choice-based operations are not new in programming language theory, as they 
first appeared in the context of computational effects (more specifically, choice operations has been first studied as nondeterministic  choices~\cite{DBLP:conf/aieeire/000161}). 
However, their integration with reward constructs as a way to 
structure machine learning systems as data-driven, decision-making processes is quite recent. 
As an example of that, the programming language \textsf{SmartChoice}~\cite{Carbune2019:PredictedVariables} integrates choice and rewards operations 
with traditional programming constructs. 
To the best of the authors' knowledge, the first work dealing with 
semantic foundations for choice-based programming languages is the recent 
work by \citet{DBLP:conf/lics/AbadiP21}. There, a modular
approach to higher-order choice-based 
programming languages is developed in terms of 
the selection monad~\cite{DBLP:conf/cie/EscardoO10a,DBLP:conf/csl/EscardoOP11,
DBLP:journals/apal/EscardoO12} and its algebraic 
operations~\cite{PlotkinPower/FOSSACS/01,
PlotkinPower2003:AlgebraicOperations,PlotkinPower/FOSSACS/02}, 
both at the level of operational and denotational 
semantics. As such, that work is the closest to ours. 

The results presented in this paper has been obtained starting from 
an investigation of possible applications 
of the aforementioned work by Abadi and Plotkin to modelling reinforcement 
learning systems. We have explained in Section~\ref{selection-monad} 
why, when dealing with reinforcement learning, moving from the selection 
to the state monad is practically beneficial.
From a theoretical point of view, 
such a shift can be seen as obtained by 
looking at comodels~\cite{DBLP:journals/entcs/PowerS04} of choice operations, 
which directly leads to view reinforcement learning algorithms as ultimately 
defining stateful runners~\cite{DBLP:journals/entcs/Uustalu15,DBLP:conf/esop/AhmanB20} 
for such operations, and thus as generic effects \cite{PlotkinPower2003:AlgebraicOperations} 
for the state monad. Following this path, we have consequently relied on handlers~\cite{PlotkinPretnar/handling-algebraic-effects-2013,
Bauer-Pretnar/Programming-with-algebraic-effects,
Pretnar2015:EFF} to define modular implementations 
of such generic effects, and thus of reinforcement learning algorithms. 
Even if handlers and algebraic operations are well-known tools in 
functional programming, to the best of our knowledge the present 
work is the first one investigating their application to 
modelling reinforcement learning systems.

Finally, we mention that applications of functional programming techniques in 
the context of machine learning similar in spirit to ones investigated in this work, 
as well as in the work by \citet{DBLP:conf/lics/AbadiP21}, 
have been proposed \cite{DBLP:conf/flops/CheungDGMR18} 
in terms of induction and abduction constructs, rather 
than in terms of choices and rewards

%% file: Conclusion.tex
In this article, we address the problem of implementing reinforcement learning 
algorithms within a functional programming language. The starting idea is to 
manage the interaction between the three involved agents (namely the learner, 
the user, and the environment) through a set of algebraic operations and 
handlers for them. This way, a certain degree of modularity is guaranteed.

The path we followed starts from practice, i.e. from the implementation of 
these ideas in a concrete language such as \EFF, and then progressively moves 
towards theory, i.e. towards the study of a paradigmatic language for effects 
and handlers in which these ideas can be formalized, thus becoming an object of 
study. Technically, the most important contribution of this work lies in 
highlighting where the state of the art is deficient with respect to the type 
of type safety and code reuse properties that one would like. 

Some ideas for future works can be sought precisely in the direction just 
mentioned and, in particular, in the definition of systems of types 
sufficiently powerful to guarantee that the algebraic operations involved are 
actually carried out in the good order, or that they allow the right level of 
polymorphism.

An alternative to powerful type systems for some safety properties could also be the use of some advanced features of programming language for abstraction. For example, we presented the abstract interface as a handler, with a learner polymorphic in the types of actions and observation. In \OCAML, an alternative could be to use modules and signatures: an abstract interface would be represented by a signature, with abstract data types for actions and observations, and additional functions to replace the operations. The learner would then have access to elements of this signature, and the user would need to implement an actual module respecting this signature in order to use the learning algorithm. This way, the fact that the learner can be written independently from the environment would come from the \OCAML\ abstraction, instead of polymorphism in a type and effect system. 